\documentclass[twoside,journey]{IEEEtran}
\usepackage{ifpdf}
\usepackage{cite}
\usepackage{url}
\usepackage[justification=centering]{caption}
\usepackage{amsmath}
\usepackage{multirow}
\usepackage{graphicx}
\usepackage{subfig}
\usepackage{booktabs}
\usepackage{longtable}
\usepackage{multicol}
\usepackage[ruled,vlined]{algorithm2e}
\usepackage{algorithmic}
\usepackage{amssymb}
\usepackage{bbding}
\usepackage{bm}
\usepackage{colortbl}
\usepackage[table]{xcolor}
\usepackage{makecell}
\usepackage{amsfonts}
\newtheorem{lemma}{Lemma}
\newtheorem{proof}{proof}

\begin{document}
	
	\pagenumbering{arabic}
	\title{Blockchain-based Federated Learning for Industrial Metaverses: Incentive Scheme with Optimal AoI}
	
		\author{Jiawen Kang, Dongdong Ye, Jiangtian Nie, Jiang Xiao, Xianjun Deng,  Siming Wang,   \\ Zehui Xiong, Rong Yu, and Dusit Niyato 
		\IEEEcompsocitemizethanks{
			J. Kang, D. Ye, S. Wang and R. Yu are with School of Automation, Guangdong University of Technology, China. J. Nie and D. Niyato are with School of Computer Science and Engineering, Nanyang Technological University, Singapore. J. Xiao and X. Deng are respectively with School of Computer Science and Technology and the Department of Cyber Science and Engineering, Huazhong University of Science and Technology, China. Z. Xiong is with Pillar of Information Systems Technology and Design, Singapore University of Technology and Design, Singapore.  }
			
		\IEEEcompsocitemizethanks{This research is supported by National Key R\&D Program of China (No. 2020YFB1807802), NSFC under grant No. 62102099, Open Research Project of the State Key Laboratory of Industrial Control Technology, Zhejiang University, China (No. ICT2022B12), and is also supported, in part, by the programme DesCartes and is supported by the National Research Foundation, Prime Minister’s Office, Singapore under its Campus for Research Excellence and Technological Enterprise (CREATE) programme,  the National Research Foundation, Singapore under the AI Singapore Programme (AISG) (AISG2-RP-2020-019), and Singapore Ministry of Education (MOE) Tier 1 (RG16/20), and  is supported by the National Research Foundation, Singapore and Infocomm Media Development Authority under its Future Communications Research \& Development Programme, and is also supported by the SUTD SRG-ISTD-2021-165, the SUTD-ZJU IDEA Grant (SUTD-ZJU (VP) 202102), and the SUTD-ZJU IDEA Seed Grant (SUTD-ZJU (SD) 202101). (\textit{Corresponding author: Zehui Xiong, zehui\_xiong@sutd.edu.sg}) }
	}

	\maketitle 
	\pagestyle{headings}
	\begin{abstract}
	The emerging industrial metaverses realize the mapping and expanding operations of physical industry into virtual space for significantly upgrading intelligent manufacturing. The industrial metaverses obtain data from various production and operation lines by Industrial Internet of Things (IIoT), and thus conduct effective data analysis and decision-making, thereby enhancing the production efficiency of the physical space, reducing operating costs, and maximizing commercial value. However, there still exist bottlenecks when integrating metaverses into IIoT, such as the privacy leakage of sensitive data with commercial secrets, IIoT sensing data freshness, and incentives for sharing these data. In this paper, we design a user-defined privacy-preserving framework with decentralized federated learning for the industrial metaverses. To further improve privacy protection of industrial metaverse, a cross-chain empowered federated learning framework is further utilized to perform decentralized, secure, and privacy-preserving data training on both physical and virtual spaces through a hierarchical blockchain architecture with a main chain and multiple subchains. Moreover, we introduce the age of information as the data freshness metric and thus design an age-based contract model to motivate data sensing among IIoT nodes. Numerical results indicate the efficiency of the proposed framework and incentive mechanism in the industrial metaverses.
	\end{abstract}
	
	\begin{IEEEkeywords}
		Metaverse, blockchain, federated learning, contract theory, incentive mechanism, age of information
	\end{IEEEkeywords}
	
	\section{Introduction}
	% 1. Metaverse notation,  
	%The metaverse enables the ubiquitous communication and mutual development between the real world and the virtual space, presenting the form of virtual and real integration and mutual influence \cite{lee2021all}. 
	% 1.1 Industrial metaverse
	As the current wave of the industrial revolution, the digital transformation of manufacturing industry has taken place.
	Different from the traditional industrial digitization to improve the physical space through digital operations, industrial metaverse creates a virtual space by transforming the physical space based on the real interaction and persistence \cite{huynh2022artificial}.
	One of the most potential applications of industrial metaverse is 3D simulation, modeling, and architectural design \cite{wang2022survey}. For instance, an open platform called Omniverse has been built by NVIDIA \cite{hummel2019leveraging}. Multi-user real-time 3D simulation and visualization of physical entities and properties are supported in a shared virtual space for industrial applications, such as automotive design. 
	%
	%The industrial metaverse is expected to meet the business needs of the growing market, such as smart parks, smart factories, and digital production lines.
	% 1.2 IIoT for industrial metaverse
	To build an industrial metaverse, the Industrial Internet of Things (IIoT) play an essential and foundational role in networking and communication.  
	The IIoT nodes collect a large amount of sensing data to bridge the virtual space and physical space, thus  providing users with a completely real, lasting, and smooth interactive experience in the industrial metaverses \cite{ning2021survey}. 
	% 1.3 Requirement of industrial metaverse
	While the industrial metaverse requires powerful hardware equipments, rich network resources and advanced Artificial Intelligence (AI) technologies as the base.
	By collaboratively applying cutting-edge technologies, such as Federated Learning (FL), blockchain and digital twins, the industrial metaverse for IIoT can significantly modernize digital operations in the current digital revolution \cite{huynh2022artificial}.
	% 1.4. Age of information for Metaverse 
	%IIoT enabled applications in the domains of product lines, environmental and industrial fault diagnosis monitoring must rely on timely and fresh information gathered from a wide plethora of IIoT sensing devices. Age of Information (AoI), the age of the data since it was generated, is a metric indicator of data freshness.
	
	% 3. Challenges
	Although industrial metaverse brings amazing changes to industrial areas, this technology is still in its infancy. There exist challenging  bottlenecks for wide deployment and future popularization:
	% 3.1 Privacy leakage;
	1) The industrial metaverse is prone to leaking IIoT nodes' privacy. Since the industrial metaverse will collect more data from IIoT nodes than ever before, the consequences will be worse than ever if IIoT nodes' data is not effectively protected \cite{wang2022survey}. Thereby, due to privacy concerns,  IIoT nodes may be unwilling to share  data with private information in industrial metaverses, which hinders the comprehensive data analysis using machine learning technologies for  intelligence improvement of the industrial metaverse. 
	% 3.2 Resource constraints; 
	%2) The tension between data explosion and  limited network resources of IIoT. With the rapid increase of IIoT nodes, unfathomably vast amounts of data are generated. Due to limited network resources, it is impossible to upload such giant data to a centralized metavese server \cite{ning2021survey}. In order to save network resources, IIoT nodes may selectively upload their data to the metaverse server, which results in data differences between virtual entities and physical entities  in industrial metaverses. 
	%The bandwidth and storage resources of IIoT are constrained. 
	%The metaverse will require a large amount of bandwidth to transmit the high-resolution content in real time. Besides, the massive amount of data and metaverse data generated by ubiquitous sensor deployments. The construction of the metaverse not only requires a huge amount of available bandwidth, but it may also compete with other applications.
	% 3.3 Age-based model / incentive mechanism 	
	2) To ensure immersive services in the industrial metaverse, fresh sensing data is significantly important to enhance the service quality of time-sensitive services.
	Since the IIoT nodes are energy-constrained, they can not respond to each data request of learning-based metaverse services \cite{corneo2019age}. As a result,  the IIoT nodes may not provide fresh data or join metaverse services without a reasonable incentive mechanism. It is still challenging how to incentivize IIoT nodes with fresh data in industrial metaverse services.
	
	% 4. Contributions 
	To address the above challenges, in this paper, we first apply federated learning and cross-chain technologies to design  a user-defined privacy-preserving framework for sensing data analysis in industrial metaverses. To improve service quality of industrial metaverse, Age of Information (AoI) is introduced as the data-freshness metric of sensing data for industrial metaverse services. We then design an age-based contract model to incentivize data sensing among IIoT nodes.   The main contributions of this paper are summarized as follows:
	\begin{itemize}
		\item We design a new privacy-preserving framework for industrial metaverse, in which IIoT nodes can customize to upload non-sensitive sensing data to the virtual space for learning-based metaverse tasks, and keep the sensitive sensing data in the physical space for privacy protection. 
		
		\item To further improve privacy protection of industrial metaverse, a cross-chain empowered FL framework is further utilized to  perform decentralized, secure, and privacy-preserving data training on both physical and virtual spaces through a hierarchical blockchain architecture with a main chain and multiple subchains. The cross-chain interaction is executed to finish secure model aggregation and update. 
		
		\item We apply Age of Information as the data-freshness metric of sensing data to optimize the time-sensitive learning tasks in industrial metaverses, along with designing  an optimal age-based  contract to incentivize data sensing among IIoT nodes for industrial metaverse.   
	\end{itemize}
	
	%-------------------------------------------------
	\begin{figure*}[t]\centering     \includegraphics[width=0.75\textwidth]{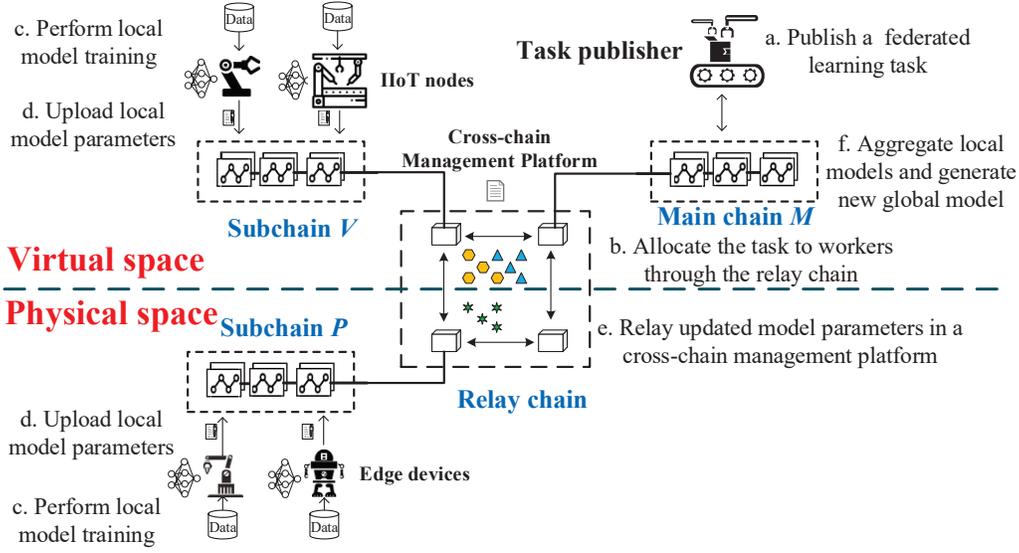}     \caption{A cross-chain empowered federated learning framework for the industrial metaverses.}     \label{system_model}     
	\end{figure*}
	%-------------------------------------------

	The rest of this paper is organized as follows.  The proposed cross-chain empowered FL for industrial metaverse is introduced in Section II. Problem formulation and solutions for  contract-based incentive mechanism are presented in Section III and Section IV, respectively. Section V shows the performance evaluation of the proposed incentive mechanism and Section VI concludes this paper finally.

	\section{Cross-chain Empowered Federated Learning Framework for Industrial Metaverses}
	
	\subsection{User-defined Privacy-preserving Training Framework}
	%1. A hierarchical blockchain framework for decentralized federated learning: a main chain as the parameter server and subchains as workers 
	As shown in Fig. \ref{system_model}, a hierarchical blockchain framework for decentralized federated learning consists of  a main chain  and multiple subchains. The main chain works as the parameter server and the subchains manage local model updates generated by IIoT nodes acting as workers. There exists a physical space and a virtual space in the framework. Here, IIoT nodes in the physical space can be sensors in smart manufacturing factories or the inspection drones in smart grids. Multiple IIoT nodes in the physical space  are managed by a subchain $P$ during training. However, due to the privacy concerns, these physical workers may be not willing to upload all privacy-sensitive data to the industrial metaverse directly. 
	
	The virtual space is established on i) collected data from the IIoT nodes in the physical space and ii) the online  generated data during node interaction or data analysis. A virtual node in the industrial metaverse is built by mapping and synchronizing the data of a physical IIoT node to the virtual space. Multiple virtual nodes as model training workers are managed by  a subchain $V$ in the virtual space. Since the physical nodes are reluctant to upload all data with privacy-sensitive information to the virtual space in the industrial metaverse. The datasets of virtual nodes are incomplete. If a learning task is only trained by the virtual nodes, the accuracy of the learning model will be decreased and the generalization ability will be poor. To this end, a user-defined privacy-preserving training framework is designed with FL for industrial metaverses, in which IIoT nodes can customize to upload non-sensitive sensing data to the virtual space for learning-based metaverse tasks or applications, and keep the sensitive sensing data locally in the physical space for strong privacy protection. The FL technology breaks data islands between the virtual space and physical space,  and enables  collaborative learning among the virtual and physical nodes.

	\subsection{Cross-chain Interaction for Decentralized FL}
	%2. Main steps: cross-chain interaction 
	To further improve privacy protection of industrial metaverse, a cross-chain empowered FL framework is further designed to  perform decentralized, secure, and privacy-preserving data training on both physical and virtual spaces through the hierarchical blockchain framework with a main chain and multiple subchains \cite{9785702}. The cross-chain interaction is executed to finish secure model aggregation and update \cite{jin2022towards,jin2021cross}. In Fig. 1, the workflow of the proposed cross-chain empowered FL framework is presented as follows.
	
	\subsubsection{\textit{Step 1: Publish a  federated learning task}} Each task publisher (e.g., a smart manufacturing company) sets up a learning task (e.g., anomaly prediction of product lines) and sends the federated learning request to the main chain $M$ (Step $a$). 
	
	\subsubsection{Step 2: Allocate the task to the workers} The main chain $M$ sends the learning task to the relay chain that  is a cross-chain management platform \cite{siriweera2021internet}. This platform is responsible for forwarding, verifying data (e.g., block data), and bridging connections among blockchains. The relay chain first verifies the task information and sends the learning task to the workers' subchain $V$ of the virtual space and the subchain $P$ of the physical space, respectively (Step $b$). 
	
	\subsubsection{Step 3: Perform a learning task in both virtual space and physical space} In the physical space, legitimate edge devices (e.g., industrial robots, IIoT sensors) can join a training task and perform local model training on their local datasets. Each dataset is generated from personal applications (e.g., product line detection services) or from the surrounding environment (e.g., sensors on smart manufacturing factories). Each node trains a given global model from its task publisher and generates local model updates (Step $c$).  Similarly, in the virtual space, the corresponding legitimate virtual nodes as workers  also participate in the learning task and  train the given global model at the same time.
	
	\subsubsection{Step 4: Aggregate model and update global model} When the nodes in the virtual and physical spaces complete the  training task, the updated local models are verified and uploaded to their subchains firstly for secure management. These models are  checked by the relay chain and transmitted to the main chain (Steps $d$, $e$). Then, the updated local models are aggregated on the main chain to generate a new global model $M$ (Step $f$). Finally, the workers download the latest global model from their subchains and train the new model for the next iteration until satisfying the given accuracy requirements. The final global model is sent back to the task publisher, and the task publisher sends the payments to the workers according to their contributions \cite{jin2021cross}. 
	
	\section{Problem Formulation}
	In this section, to incentivize data sharing among IIoT nodes for time-sensitive FL tasks, we first introduce Age of Information (AoI) as the metric to evaluate the data freshness.  We then formulate the utility functions of both workers and service provider (i.e., the task publisher) in industrial metaverse.
	
	We consider a mixed reality based remote monitoring case  as an example of industrial metaverse scenarios with a service provider and $M$ workers. 
	The service provider acting as task publisher motivates $M$ workers to participate in learning tasks.  The average time of a global model iteration in the cross-chain empowered FL consists of  three parts: 1) the average time of completing a global model iteration of federated learning (denoted as $t_u$);  2) the average time of completing a consensus process for a global model iteration among blockchains (denoted as $t_c$);  3)  the average time of collecting and processing the data for model training (denoted as $c_m t, c_m \in \mathbb{N}$, and $t = t_u + t_c$).
	We consider that the federated learning is synchronous, and $t_u$ is the same for all the workers \cite{lim2020hierarchical, lim2020information}. 
	$t_c$ is the same for each global model iteration because of using the same relay chain \cite{shen2020blockchain}. 
	For each worker $m$, the time of collecting and processing the data for model training is a constant \cite{lim2020hierarchical, lim2020information}. 
	
	% Following \cite{zhang2018towards}, the average time of completing a global iteration of federated learning in the blockchain is denoted as $t = t_u + t_c$ where $t_u$ is the average time of completing a global iteration of federated learning and $t_c$ is the average time of completing a consensus process for a global iteration in the blockchain.
	% Here, we consider a case in which the federated learning is synchronous, and thus $t_u$ is the same for all the workers. The case also appears in \cite{lim2020hierarchical, lim2020information}
	% According to \cite{}, we assume that $t_c$ is the same for each global iteration. 
	% Moreover, each duration $T$ can be represented in terms of instances of $t$, and 
	% For worker $m$, the time of collecting and processing the data for model training is a constant which is denoted as $c_m t, c_m \in \mathbb{N}$. 
	
	\subsection{Age of Information and Service Latency for Industrial Metaverse}
	Recently, Age of Information (AoI) has emerged as a metric to quantify  information freshness at the destination, which is a promising metric to improve performance of time-critical applications and services, especially for industrial monitoring and sensor networks \cite{lim2020information}. In FL based remote monitoring, without loss of generality, we consider a model training request arrives at the beginning of each epoch. We focus on the AoI and service latency of the FL with data caching buffer in workers \cite{lim2020hierarchical, lim2020information}.
	Worker $m$ periodically updates its cached data. The periodic interval $\theta_m$ is independent of the period in which the request arrives, which is denoted as 
	\begin{equation}
		\begin{split}
			\theta_m = c_m t + a_m t, a \in \mathbb{N},
		\end{split}
	\end{equation}
	% and responds to the arriving request through local model training on the cached data. 
	where $a_m$ is the duration time from finishing data collection to the beginning of the next phase of  data collection, which includes multiple time periods, {such as service time period or idle time period \cite{zhou2021towards}}.
	% Moreover we assume the ideal cache, i.e., each worker is able to cache and update all the data, and leave the design of specific caching schemes for our future works.
	
	% \textcolor{red}{I'm not sure, Following the characteristics of the Poisson process, the probability of a request arrival is identical across periods and is given by $\frac{1}{T}$.} 

	According to \cite{lim2020information}, if an  FL training request is raised at the $z^{th}$ period belonging to the data collection phase, the service latency is $c_m t + t - (z-1) t$. 
	Otherwise, if a request is raised at any remaining time period in the update cycle, the service latency is $t$. 
	% at the $z^{th}$ period during the data collection, the service latency is $c_m t + t - (z-1) t$. 
	% Otherwise, if the request arrives at any of the remaining period within an update cycle, the service latency is $t$. 
	Thus, the average service latency $D_m$ \cite{lim2020information} of the blockchain-based FL for worker $m$ is 
	\begin{equation} \label{eqation:D_s}
		\begin{split}
			\overline{D}_m = \frac{c_m}{ c_m + a_m } \left[ \frac{c_m t}{2} (c_m + 3) \right]  + \frac{a_m t}{ c_m + a_m }. \\
		\end{split}
	\end{equation}

	If the FL request is raised during the data collection phase or at the beginning of phase $(c_m + 1)t$, the AoI is $t$ which is the minimum value. 
	Otherwise, if a request is raised at period $lt$, the AoI will be $[l - (c_m + 1) + 1]t$, where $l \geq (c_m + 2)t$. 
	% the AoI of the data is at the minimum
	% t only if the request comes during the data collection phase, or at the beginning of phase $(c_m + 1)t$.
	% Otherwise, the AoI for a request that arrives at period $lt$ will be $[l - (c_m + 1) + 1]t$, where $l \geq (c_m + 2)t$, i.e., the periods after data collection has been completed. 
	Thus, the average AoI is denoted as 
	\begin{equation}\label{eqation:A_s}
		\begin{split}
			\overline{A}_m = \frac{t}{ c_m + a_m } \left[ c_m + 1 + \frac{(a_m - 1)(a_m + 2) }{2}   \right].  \\
		\end{split}
	\end{equation}

	When $t$ is fixed, the update cycle $\theta_m$ is affected by $c_m$ or $a_m$. In this paper, we consider a general case that has an adjustable update phase and a fixed idle phase of workers. In other words, when $a_m = a $ is fixed, we have $c_m = \frac{\theta_m }{t} - a $. We replace $c_m$ with $\theta_m$. Thus, $\overline{D}_m$ and $\overline{A}_m$ are simplified to
	\begin{equation} \label{eqation:D_s_2}
		\begin{split}
			\overline{D}_m(\theta_m) & = \frac{c_m}{ c_m + a_m } \left[ \frac{c_m t}{2} (c_m + 3) \right]  + \frac{a_m t}{ c_m + a_m }  \\
			& = \frac{(\theta_m - at)^3}{2t\theta_m} + \frac{3 (\theta_m - at)^2}{2\theta_m} + \frac{at^2}{\theta_m},  \\
		\end{split}
	\end{equation}
	and
	\begin{equation}\label{eqation:A_s_2}
		\begin{split}
			\overline{A}_m(\theta_m) & = \frac{t}{ c_m + a_m } \left[ c_m + 1 + \frac{(a_m - 1)(a_m + 2) }{2}   \right]  \\
			& =  \frac{t \theta_m }{\theta_m -at} + \frac{t^2}{\theta_m -at} \left(  \frac{a^2 - a }{2} \right) .  \\
		\end{split}
	\end{equation}
	The case is established under $\theta_m < at, \forall m \in M$.
	Since $\theta_m = a_m t + c_m t$, $\theta_m > a_m t $ always holds. 
	When $\theta_m > at $, $\overline{D}_m(\theta_m)$ is a convex function with respect to $\theta_m$.
	When $\theta_m > at $ and $a>1$, $\overline{A}_m(\theta_m)$ is also a convex function with respect to $\theta_m$. 
	% When $\theta_m > at $, $\overline{D}_m(\theta_m)$ is a convex function with respect to $\theta_m$.
	% When $\theta_m > at $, $\overline{A}_m(\theta_m)$ is a convex function with respect to $\theta_m$. 

	\subsection{Worker Utility}

	The utility of worker $m$ is the difference between the received monetary reward $R_m$ and its cost $C_m$ of participating in FL training task.
	\begin{equation}
		\begin{split}
			U_m = R_m - C_m.  \\
		\end{split}
	\end{equation}
	Referring to \cite{zhou2021towards}, $C_m$ is defined as 
	\begin{equation}
		\begin{split}
			C_m = \frac{\delta_m}{ \theta_m}, \\
		\end{split}
	\end{equation}
	where $\delta_m$ is the unit update cost and is related to data collection, computation, transmission, and consensus \cite{zhang2018towards}.

	Thus, the utility of worker $m$ becomes
	\begin{equation}
		\begin{split}
			U_m = R_m - \frac{\delta_m}{ \theta_m}.  \\
		\end{split}
	\end{equation}
	Since the service provider does not know the unit update cost of  each worker, it can sort the workers into discrete types according to statistical distributions of the worker types from historical data, and thus optimize the expected utility of the service provider. Specifically,  we divide the workers into different types and denote the $n$-th type worker as $\delta_n$.
The workers can be classified into a set $\Delta= \left\{ \delta_n:1 \leq n \leq N \right\}$ of $N$ types.
	In a non-decreasing order, the workers' types are sorted as:
	$\delta_1 \geq \delta_2 \geq \dots \geq \delta_N$. For the convenience of explanation, the worker with type $n$ is called the type-$n$.
	The probability that a worker belongs to a certain type-$n$  is $Q_{n}$ and we have $\sum_{ n \in N } Q_{n} =1$.
	Thus, the utility of the type-$n$ worker can be rewritten as
	\begin{equation}
		\begin{split}
			U_n = R_n - \frac{\delta_n }{ \theta_n}.  \\
		\end{split}
	\end{equation}
	To simplify the description, we define the update frequency as $f_n = \frac{1}{\theta}$, $\gamma_n = \frac{1}{\delta_n}$. Thus, the the utility of the type-$n$ worker can be rewritten as 
	\begin{equation}
		\begin{split}
			U_n = R_n - \frac{f_n }{ \gamma_n}.  \\
		\end{split}
	\end{equation}
	The workers' types, i.e., $\delta_1 \geq \delta_2 \geq \dots \geq \delta_N$ is rewritten as $\gamma_1 \leq \gamma_2 \leq \dots \leq \gamma_N$. 
	
	\subsection{Service Provider Utility}
	Note that a large AoI and a large service latency lead to a bad immersive experience and also reduce the service provider's satisfaction in industrial metaverses \cite{jiang2021reliable}.  
	Thus, the service provider's satisfaction function obtained from the type-$n$ worker is defined as follows,
	\begin{equation}
		\begin{split}
			G_n = \beta \log\left(  g_n(f_n)  \right),
		\end{split}
	\end{equation}
	where $\beta$ is the unit profit for the satisfaction, and $g_n$ is the performance obtained from the type-$n$ worker. Referring to \cite{zhou2021towards}, $g_n$ is defined as 
	\begin{equation}\label{eqation:g_s}
		\begin{split}
			g(f_n) = \alpha_n (K - \overline{A}_n) + (1 - \alpha_n) ( H - \overline{D}_n ),
		\end{split}
	\end{equation}
	where $\alpha_n$ is the preference factor on AoI for the service provider to the type-$n$ worker. The larger $\alpha_n$ means the service provider prefers the AoI more.   $K$ and $H$ are the maximum tolerant AoI and service latency, respectively.
	According to the types of workers, the utility of the service provider is 
	\begin{equation}
		\begin{split}
			U_s = \sum_{ n \in N} M Q_n ( G_n  - R_n).
		\end{split}
	\end{equation}

	\section{ Optimal Contract Design}\label{MDCD}
	\subsection{Contract Formulation}
	Noted that the types of workers are private information that is not visible to the service provider, namely, there exists information asymmetry between workers and the service provider. Under the information asymmetry, the service provider uses contract theory to find out the best workers.
	Here, the service provider is the leader for designing contracts, and each worker selects the best contract item according to its type. 
	The contract item can be denoted as $\Phi = \left\{ (\gamma_{n}, f_{n}, R_{n}), n \in N   \right\}$, with item $(\gamma_{n}, f_{n}, R_{n})$ for type-$n$ worker. 
	In order to ensure that each worker automatically selects the contract item designed for its specific type, the feasible contract must satisfy the following   Incentive Compatibility (IC) constraint:
	
	\begin{equation}
		\begin{split}\label{IC1}
			R_{n} - \frac{  f_n }{ \gamma_{n} }  \geq  R_{i} - \frac{  f_i }{ \gamma_{n} }  , \forall n, i \in N, 
		\end{split}
	\end{equation}
      and the Individual Rationality (IR) constraint:
	
	\begin{equation}\label{IR2}
		\begin{split}
			R_{n} - \frac{  f_n }{ \gamma_{n} }  \geq 0, \forall n \in N.
		\end{split}
	\end{equation}
	With the IC and IR constraints, the problem of maximizing the expected utility of the service provider is formulated as
	\begin{eqnarray}\label{problem1}
		\begin{split}
			\textbf{Problem 1:} & \quad \max_{ \bm{f}, \bm{R}} U_s \\
			\text{s.t.} & \quad \text{IC Constraints in} (\ref{IC1}), \text{IR Constraints in} (\ref{IR2}),
		\end{split}
	\end{eqnarray}
	where $\bm{f} =[f_{n}]_{1 \times N}$, and $\bm{R} =[R_{n}]_{1 \times N}$. 
	There are $N$ IR constraints and $N(N-1)$ IC constraints, making it quite difficult to  solve \textbf{{Problem 1}} directly (\ref{problem1}).

	\begin{algorithm}[t] 
		\caption{Finding Optimal Contract}
		\label{alg:r2p}
		
		\For{$n = 1, \dots, N$}
		{
			
			Initialize the iteration index $z = 0$, the step size $\varphi$, the empty vector $\bm{\mathcal{U}}_{s,n}$, $f_{min}  = f_n^{z} = 10^{-5}$ and $f_{max}$ \\
			
			\While{$f_n^{z} < f_{max}$}
			{
				
				Compute $U_{s,n}(f_n^{z})$ \\
				Set $\bm{\mathcal{U}}_{s,n}(z) = U_{s,n}(f_n^{z})$ \\
				$f_n^{z} = f_n^{z} + \varphi$ \\
				$z = z + 1$ \\
				
			}
			
			Get the optimal update frequency $f_n^{\star}$ for type-$n$ worker  by using the index of the maximum value in $\bm{\mathcal{U}}_{s,n}$  \\
			
		}

		Get the vector of the optimal update frequency $\bm{f}^{\star '} = \left\{ f_1^{\star}, \dots, f_n^{\star}, \dots, f_N^{\star} \right\}$

		\If{ $\bm{f}^{\star '}$ does not satisfy the monotonicity condition}
		{
			Use "Bunching and Ironing" algorithm \cite{gao2011spectrum} to adjust $\bm{f}^{\star '}$ and output $\bm{f}^{\star}$
		}

		\If{ $\bm{f}^{\star '}$ satisfies the monotonicity condition}
		{
			$\bm{f}^{\star}$ = $\bm{f}^{\star '}$
		}

		\For{$n = 1, \dots, N$}
		{
			Based on Eq. (\ref{eq:reward}), we will compute the optimal reward $R_n^{\star}$    
			
		}
		Get the vector of the optimal reward $\bm{R}^{\star '} = \left\{ R_1^{\star}, \dots, R_n^{\star}, \dots, R_N^{\star} \right\}$
		
		Output $\bm{f}^{\star}$ and $\bm{R}^{\star}$
		
	\end{algorithm}

	\subsection{Optimal Contract Solution}
	We simplify the \textbf{{Problem 1}} (\ref{problem1}) through the following steps. 
	Firstly, the IR and IC constraints can be reduced by \textbf{\textit{Lemma \ref{lemma:three_conditions}}}.
	Then, the service provider's optimal reward $R_n^{*}$ is derived as shown in \textbf{\textit{Lemma \ref{lemma:optimal_reward}}}. 
	Finally, the optimal contract $f_n^{*}$ can be found. More specifically,

	\begin{lemma}\label{lemma:three_conditions}
		With weakly incomplete information, a feasible contract must satisfy the following three conditions:
		
		\noindent $(a.1)  R_{1} - \frac{  f_1 }{ \gamma_{1} } \geq 0$ ;\\
		\noindent $(a.2) R_{1} \leq \dots \leq R_{N}$ and  $f_{1} \leq \dots \leq f_{N}$ ; \\
		\noindent $(a.3) R_{n} - \frac{  f_n }{ \gamma_{n} } \geq  R_{n-1} - \frac{  f_{n-1} }{ \gamma_{n} }, \forall n \in \left\{   2, \dots, N \right\}$;\\
		\noindent $(a.4)  R_{n} - \frac{  f_n }{ \gamma_{n} } \geq R_{n+1} - \frac{  f_{n+1} }{ \gamma_{n} }, \forall n \in \left\{   1, \dots, N-1 \right\}$.\\	
	\end{lemma}
	\begin{proof}
		Please refer to \cite{ding2020optimal}. 
	\end{proof}
	
	Constraint $(a.1)$ ensures that the utility of each worker receiving the contract item of its type is non-negative, which corresponds to the IR constraints. 
	Constraints $(a.2)$, $(a.3)$ and $(a.4)$ are related to IC constraints.
	Constraint $(a.2)$ shows that a worker type with a lower cost can provide the service provider with a lower update cycle, while getting more rewards from the service provider. 
	Constraints $(a.3)$ and $(a.4)$ show that the IC conditions can be reduced as the Local Downward Incentive Compatibility (LDIC) and the Local Upward Incentive Compatibility (LUIC) with monotonicity, respectively.

	Based on \textbf{\textit{Lemma \ref{lemma:three_conditions}}}, we can obtain the optimal rewards for any update cycle by the following \textbf{\textit{Lemma \ref{lemma:optimal_reward}}}:
	\begin{lemma}\label{lemma:optimal_reward}
		For a feasible set of update frequency $\bm{f}$ satisfying $f_1 \leq f_{2} \leq \dots \leq f_n \leq \dots \leq f_{N}$, we can obtain the optimal reward by the following formula
		\begin{equation}\label{equation:R0}
			R_n^{*} = \left\{ \begin{array}{l}
				\frac{  f_1 }{ \gamma_{1} } , n =1 \\
				R_{n-1} + \frac{ f_n }{  \gamma_{n} } - \frac{ f_{n-1} }{  \gamma{n} } ,\text{otherwise}.\\
			\end{array} \right.
		\end{equation}
	\end{lemma}
	\begin{proof}
		Please refer to \cite{ding2020optimal}.
	\end{proof}
	The optimal reward in (\ref{equation:R0}) is rewritten as 
	\begin{equation}\label{eq:reward}
		R^{*}_n = \frac{   f_1 }{ \gamma_{1} } + \sum^{n}_{i=1} \Delta_{i}, n \in N, 
	\end{equation}
	where $\Delta_{1}=0$ and $\Delta_{i} = \frac{  f_i }{  \gamma_{i} } - \frac{ f_{i-1} }{  \gamma_i }$,$i= 2,\dots, N$.
	We substitute the optimal rewards (\ref{eq:reward}) into the service provider's utility and we get the service provider's utility in terms of $\bm{f}$, which provide the analysis of the optimal update cycle $\bm{f}$. Thus, the optimization problem (\ref{problem1}) is rewritten as
	\begin{equation}\label{problem2}
		\begin{split}
			\textbf{Problem 2:} & \quad \max_{\bm{f}}  \quad  U_{s} \\
			\textrm{s.t.} & \quad f_{1} \leq \dots \leq f_{N},
		\end{split}
	\end{equation}
	where $U_{s} = \sum_{n \in N} U_{s,n} = \sum_{n \in N} M \left(  Q_n G_n - b_n f_n \right) $. 
	Here, $b_{n}=\frac{Q_{n}}{\gamma_n} +\left(\frac{1}{{\gamma}_{n}}-\frac{1}{{\gamma}_{{n}+1}}\right)\sum_{j=n+1}^{N} Q_{j}$ with $ n<N $, $b_{n}=\frac{ Q_{n} }{{\gamma}_{n}}$ with $n=N$.
	
	Noted that $U_{s}$ is not a concave function. Thus, we propose a greedy algorithm to find the optimal contract. The detail of the greedy algorithm is as shown in \textbf{Algorithm \ref{alg:r2p}}. 
	% https://zhuanlan.zhihu.com/p/50479555 判断算法复杂度
	In order to deploy the algorithm efficiently, we could evaluate its computational complexity, and find that its computational complexity is
	$\mathcal{O}(N \log \left(   \frac{f_{max} - f_{min}}{\varphi} \right) )$. The result indicates that the computational resource consumed by \textbf{Algorithm \ref{alg:r2p}} is at a moderate level, so it is practical to  adopt the proposed algorithm to the blockchain-based FL applications.

\begin{table}\label{parameter}
	\renewcommand{\arraystretch}{1}
	\caption{ Parameter Setting in the Simulation. }\label{table} \centering \tabcolsep=5pt
	\begin{tabular}{p{4cm}<{\raggedright}|p{3cm}<{\raggedright}}	 	
		\hline		
		\textbf{Parameter} & \textbf{Setting}\\	
		\hline
		Time taken for completing a global iteration and the consensus process $t$ &  2 s \\	
		\hline
		% 		Unit update cost $\delta$ &   \\	
		% 		\hline	
		Unit of time taken for data collection and process $c$  &  [1,15]  \\	
		\hline
		Duration from finishing data collection to the beginning of the next data collection phase $a$  & [1,15]  \\
		\hline
		Unit profit for the satisfaction $\beta$  &  20  \\	
		\hline		
		Maximum tolerant AoI $K$ &  200 s \\	
		\hline		
		Maximum tolerant latency $H$ &  50 s\\
		\hline	
	\end{tabular}\label{table3}
\end{table}
	
	\begin{figure}[htbp]
		\centering
		\includegraphics[width=0.48\textwidth]{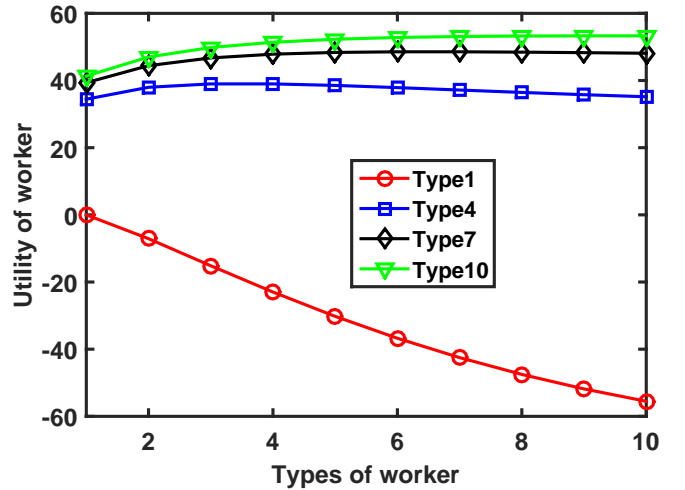}
		\caption{Utilities of workers under different types.}  \label{fig:case2_type}
	\end{figure}

	\begin{figure*}[htbp]
		%\begin{figure*}[htbp]
		\centering
		\subfloat[ Utility of the service provider.]
		{\includegraphics[width=0.45\textwidth]{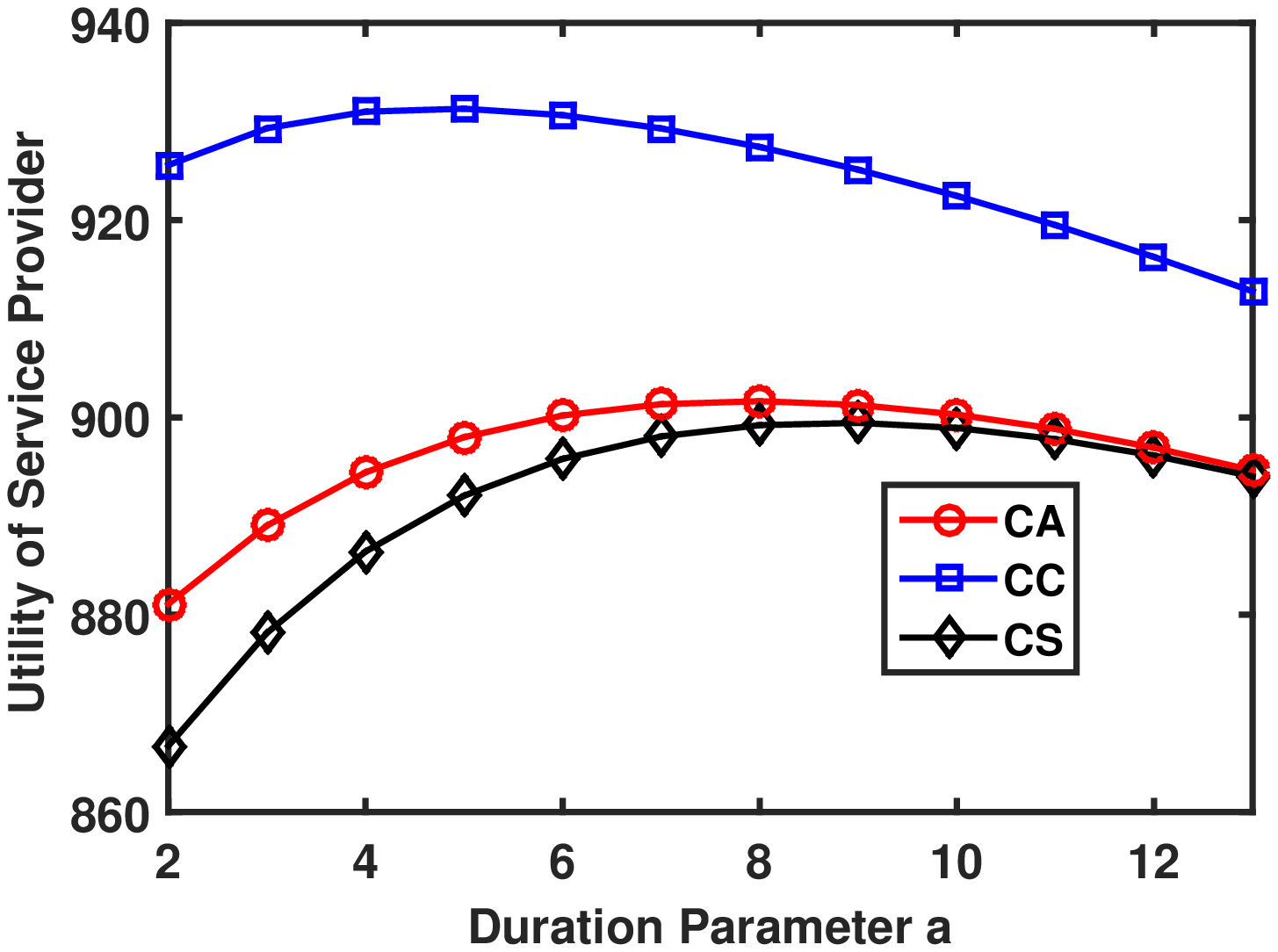}\label{fig:case2_u_s}}
		\subfloat[  Utilities of workers.]
		{\includegraphics[width=0.45\textwidth]{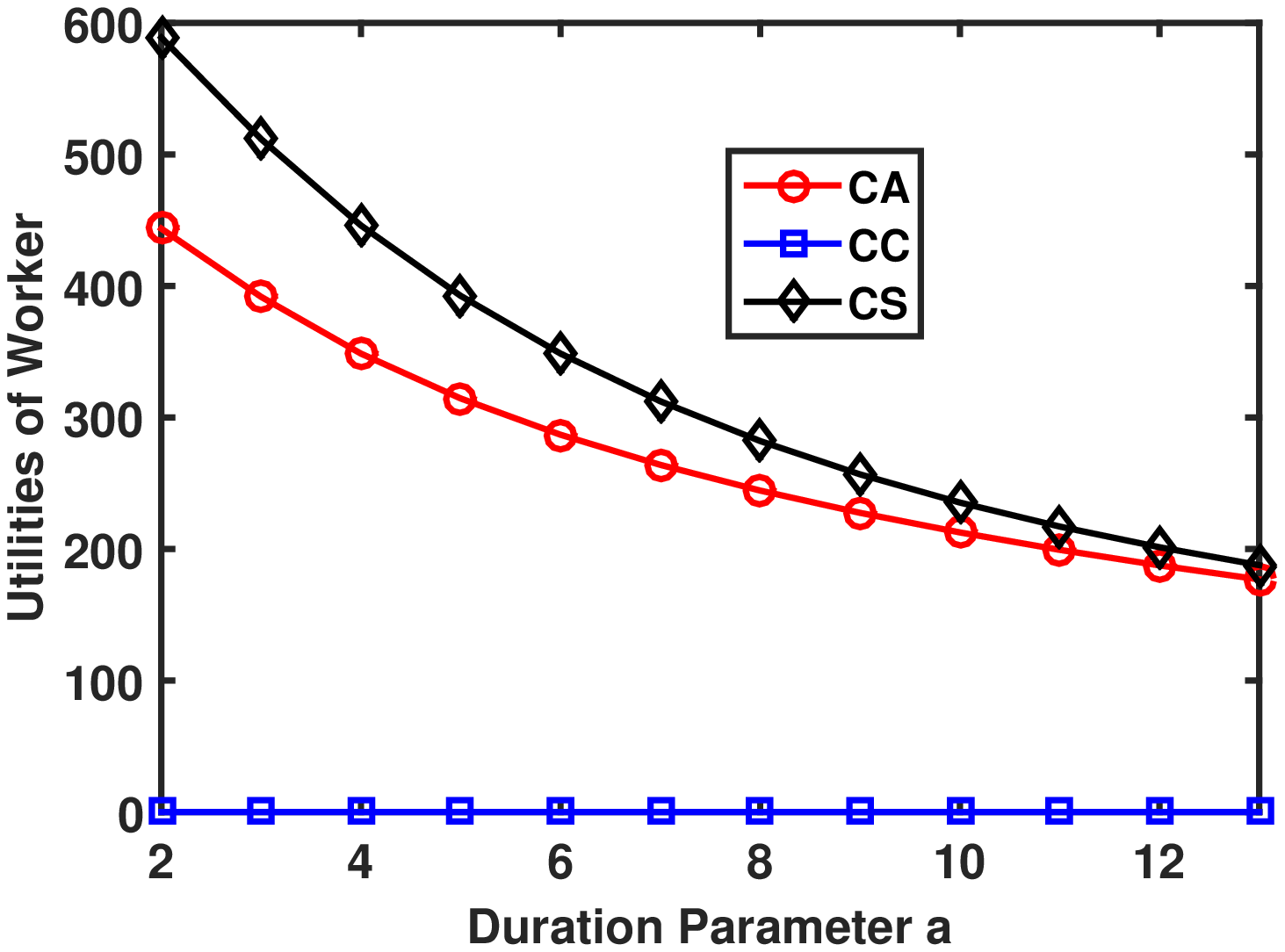}\label{fig:case2_u_n}}
		\caption{ Utility vs. Duration Parameter $a$.}\label{fig:U_parameter_c}
	\end{figure*}

	\begin{figure*}[htbp]
		\centering
		\subfloat[ Number of update cycles.]
		{\includegraphics[width=0.5\textwidth]{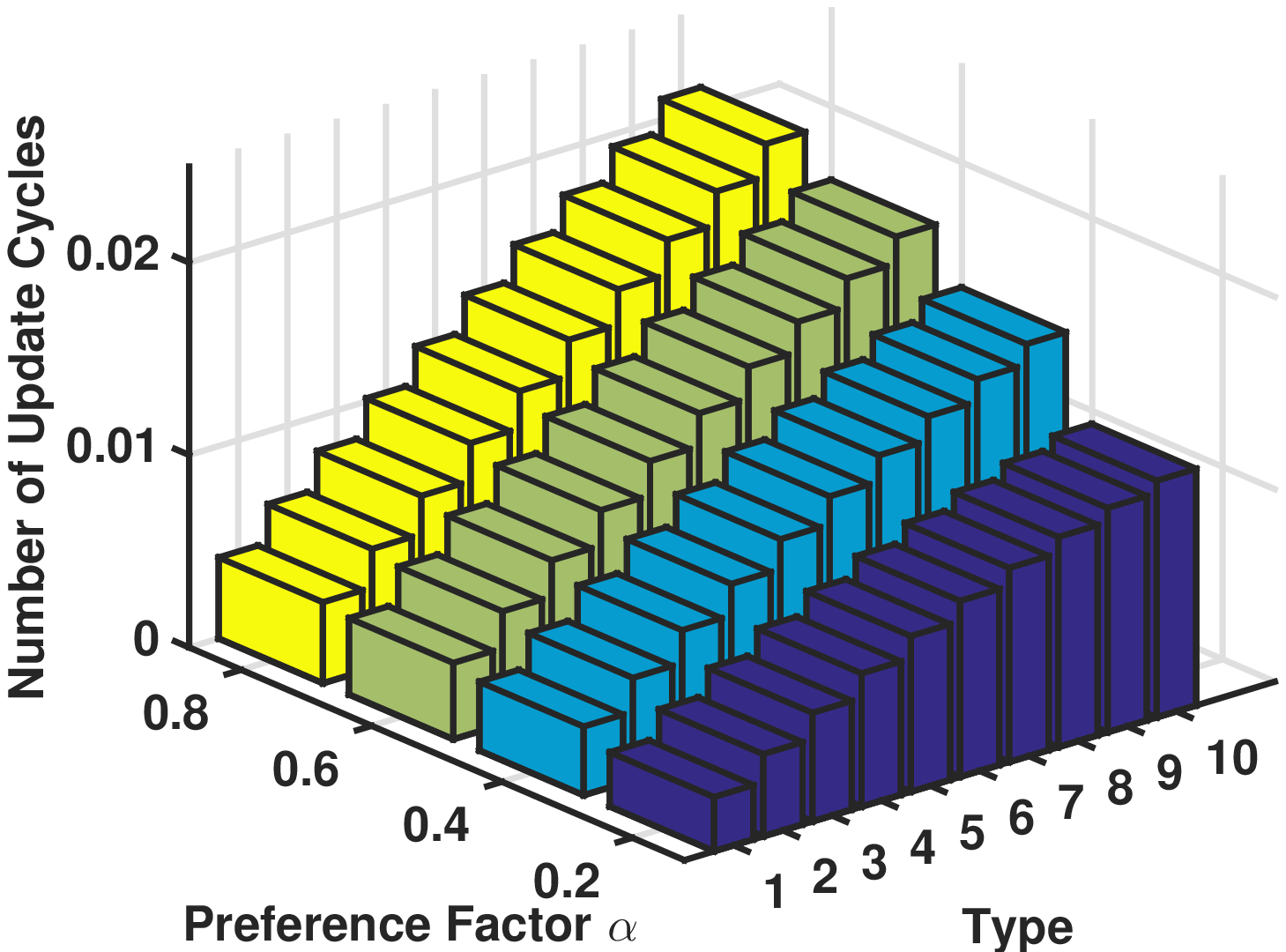}\label{fig:case2_u_s}}
		\subfloat[  Reward.]
		{\includegraphics[width=0.5\textwidth]{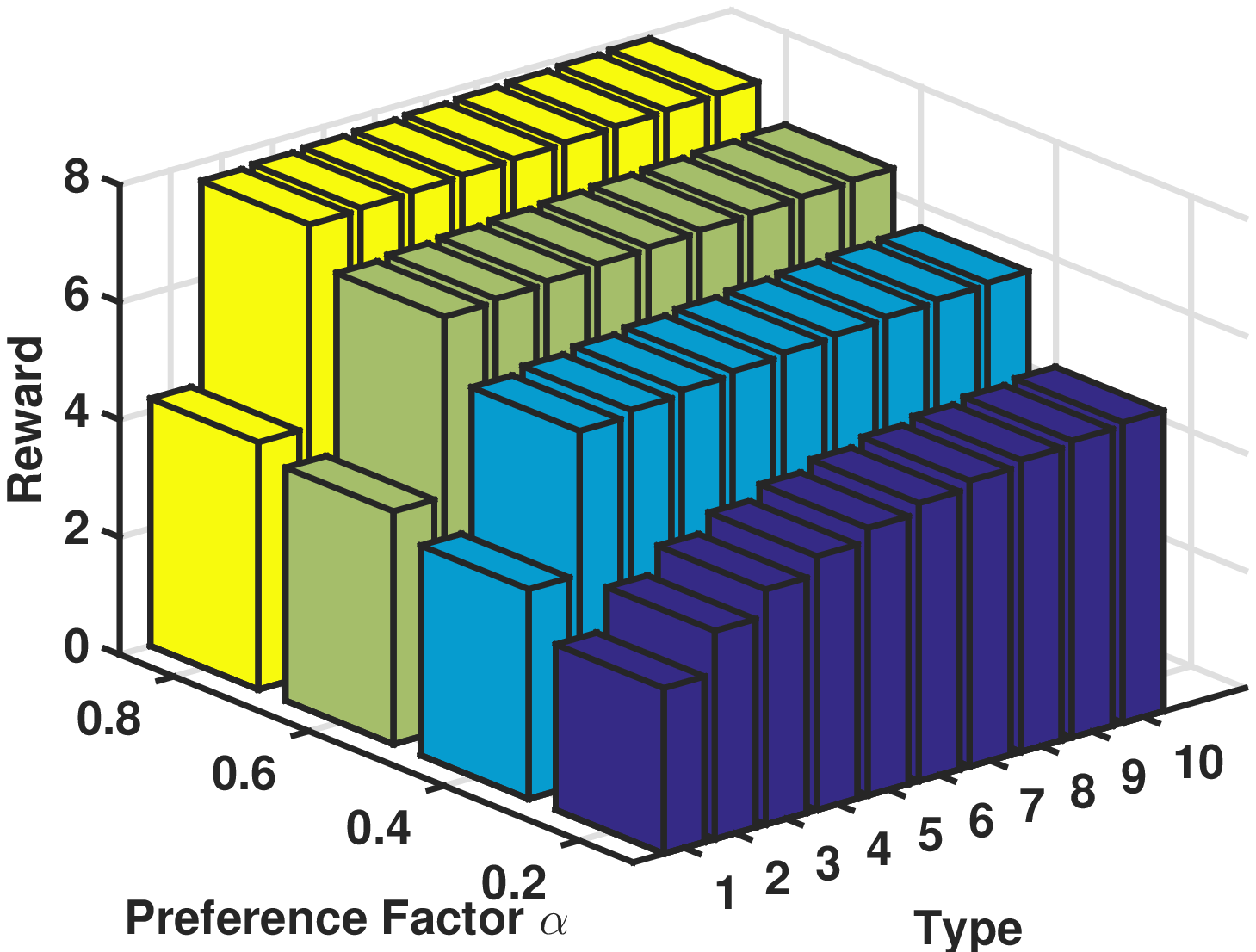}\label{fig:case2_u_n}}
		
		\subfloat[  Utility of the service provider.]
		{\includegraphics[width=0.5\textwidth]{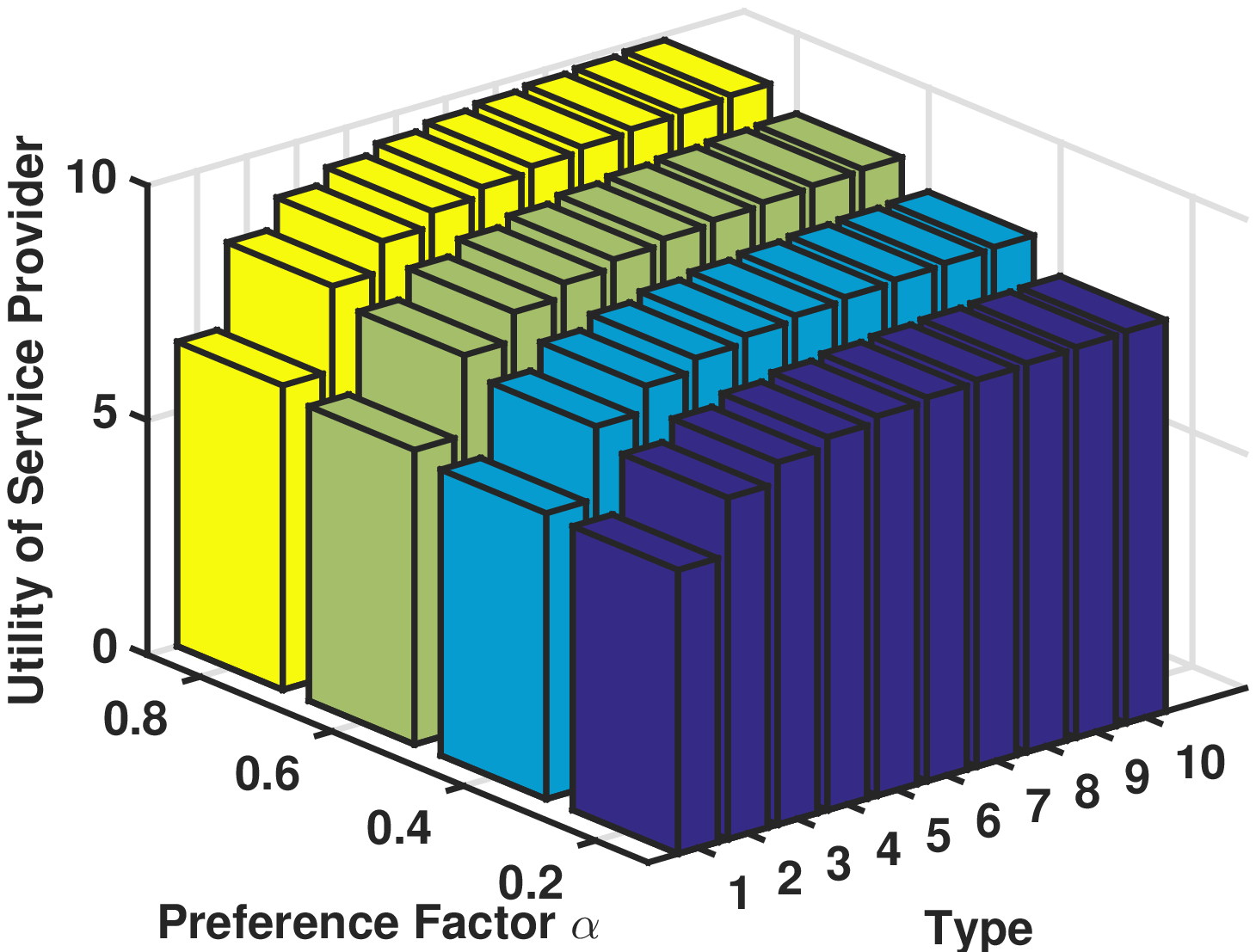}\label{fig:case2_u_n}}	
		\subfloat[  Utilities of workers.]
		{\includegraphics[width=0.5\textwidth]{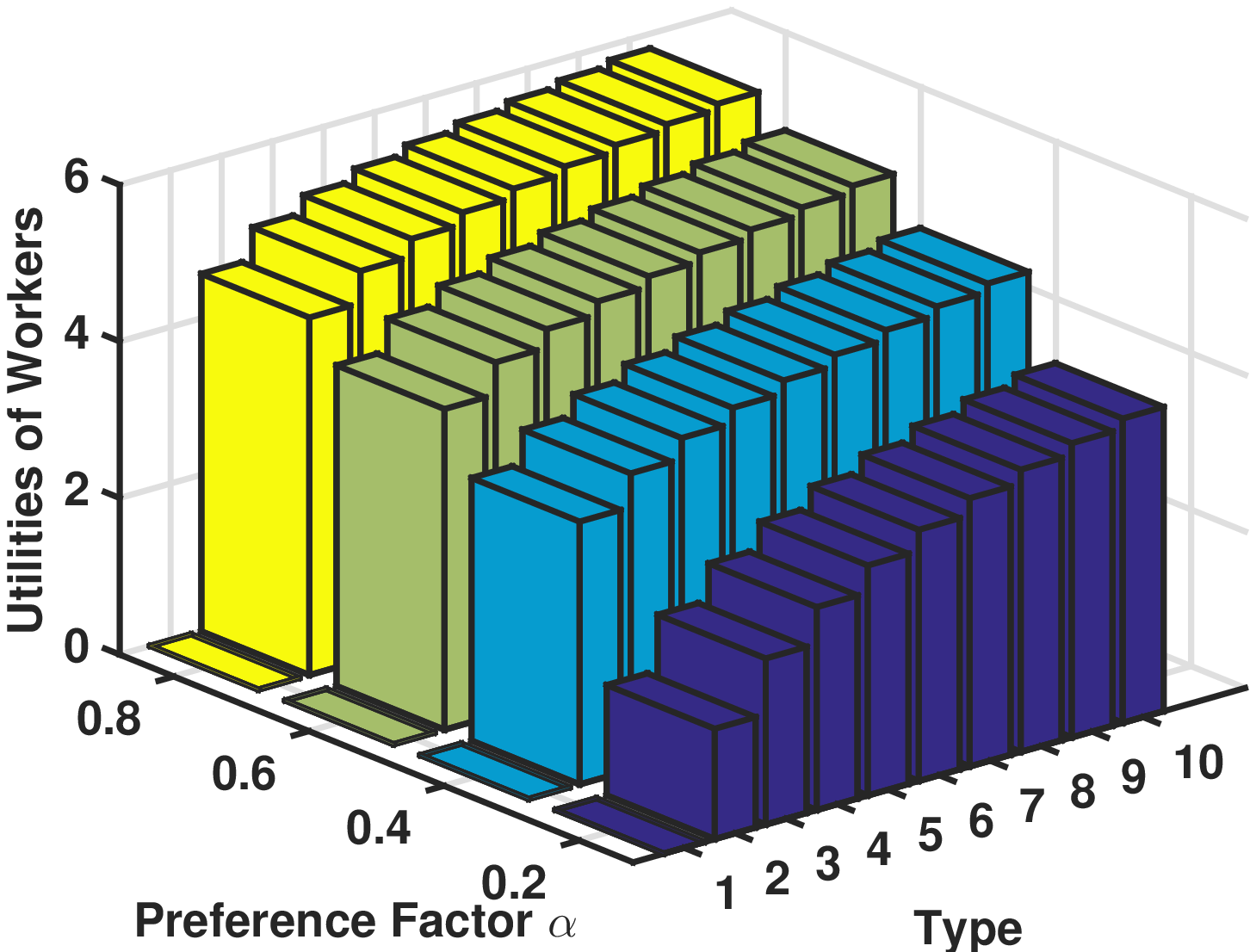}\label{fig:case2_u_n}}
		\caption{Preference Factor $\alpha$ vs. Contract Items and Utilities.}\label{fig:Parameter_alpha_fixed_c}
	\end{figure*}

\section{Numerical Results}\label{NR}
We implement the proposed cross-chain empowered FL framework by using PySyft and the Fisco Bcos blockchain  with a cross-chain platform named ``WeCross". For the simulation setting of the proposed incentive mechanism, we consider $M = 20$ workers and the type-$n$ follows the uniform distribution that is distributed in the range of $[0.001,0.01]$. Similar to  \cite{lim2020information,zhou2021towards, zhang2018towards}, the main parameters are listed in Table \ref{table3}. We compare the proposed Contract-based incentive mechanism with Asymmetric information (CA)	
with  other incentive mechanisms: i) Contract-based incentive mechanism with Complete information  (CC) that the data cost types of workers are  known by the service provider, ii) Contract-based incentive mechanism with Social maximization (CS) \cite{xiong2020multi} that the service provider maximizes the social welfare under information asymmetry \cite{9739801}.

	Fig. \ref{fig:case2_type} shows the feasibility of contract items.
	We can observe that, the utilities of  workers are increasing with the increasing types of the workers. As expected, each worker  selects the contract item corresponding to  its own type that exactly maximizes its own utility. For example, the type-4 workers obtain the optimal utility only when they choose the type-4 contract. The utilities of the workers with higher types are larger than those with lower types. Therefore, the results in Fig. \ref{fig:case2_type} validate that the contract items meet the IC and IR conditions in the proposed schemes \cite{9739801}.

	Fig. \ref{fig:U_parameter_c} shows the utilities of the service provider and the workers in terms of duration parameter $a$ under different incentive mechanisms. 
	From Fig. \ref{fig:U_parameter_c}, we can observe that the utility of the service provider first increases and then decreases  with the increasing duration parameter $a$. It means that there exists an optimal duration parameter $a$ for the workers. For a given parameter $a$, the service provider has the best utility under the CC mechanism with complete information, in which the service provider can design contract items to only maximize its benefit and set the benefit of workers as zero, which  is not so practical in the industrial metaverse scenarios. The proposed CA mechanism is practical in the real world and  has better utility than that of the CS mechanism. Meanwhile, the workers obtain the optimal utilities under the CS mechanism. And the  utilities of the workers under the CA mechanism are better than those under  the CC mechanism.

	Fig. \ref{fig:Parameter_alpha_fixed_c} shows the effect of preference factor $\alpha$ on contract items and utilities under adjustable update phase with fixed idle phase of workers.
	For a fixed type, the increasing $\alpha$ brings a larger number of update cycles and a larger reward. 
	The reason is that the growing $\alpha$ indicates that the service provider prefers to small AoI than service delay, which causes the increase of the update phase. The increase of the update phase leads to an improved reward and a larger number of update cycles. As a result, the utilities of the workers also grow.
	Moreover, the increasing $\alpha$ improves the utility of the service provider as well. The reason is that $\alpha$ has a linearly increasing relationship with the utility of the service provider when $K - \overline{A}_n > H - \overline{D}_n, \forall n \in N $. Conversely, $\alpha$ has a linearly decreasing relationship with the utility of the service provider when $K - \overline{A}_n < H - \overline{D}_n, \forall n \in N $.

	\section{Conclusion}
	In this paper, we studied data privacy protection issues and incentive mechanism design for the industrial metaverse. We proposed a privacy-preserving framework for data training by federated learning on both virtual space and physical space of the industrial metaverse. Cross-chain technology is utilized to design a decentralized FL architecture with a main chain and multiple subchains for secure model training. Furthermore, to improve service quality of time-sensitive learning tasks, we introduce age of information as the metric of data freshness and design an AoI based contract theory model for incentivizing  IIoT nodes contributing fresh sensing data. Numerical results show the efficiency of the proposed framework and incentive mechanism for the industrial metaverses.

	\bibliographystyle{IEEETran}
	\bibliography{ref_wu}
	
\end{document}